\documentclass{ifacconf}
\usepackage{graphicx}
\usepackage{amsmath}
\usepackage{amssymb}
\usepackage{mathtools}
\usepackage{amsfonts}
\usepackage{amscd}
\usepackage{bm}
\usepackage{tikz}
\usetikzlibrary{positioning, arrows.meta}
\usepackage{enumitem}

\usepackage{natbib}

\newcommand{\E}{\mathbb{E}}
\newcommand{\N}{\mathcal{N}}
\newcommand{\tr}{\operatorname{Tr}}

\definecolor{sendercolor}{RGB}{31,119,180}   
\definecolor{receivercolor}{RGB}{214,39,40}  

\begin{document}
\begin{frontmatter}

\title{Strategic Gaussian Signaling under Linear Sensitivity Mismatch\thanksref{footnoteinfo}} 

\thanks[footnoteinfo]{This work was funded by the CNRS MITI project BLESS.}

\author[CRAN,UIR]{Hassan Munif}, 
\author[CRAN]{Vineeth S. Varma}, 
\author[CRAN]{Samson Lasaulce}
\address[CRAN]{Universit\'e de Lorraine, CNRS, CRAN, F-54000 Nancy, France.}
\address[UIR]{Universit\'e Internationale de Rabat, TICLab, Rabat, Morocco.}

\begin{abstract}
We analyze Stackelberg Gaussian signaling games where the encoder and decoder have a linear sensitivity mismatch. Unlike the standard additive-bias model, a sensitivity mismatch means the encoder prefers the decoder to track a linear transformation of the state rather than a shifted one. We derive the equilibrium structure for both noiseless (cheap-talk) and noisy signaling channels. In the noiseless case, the equilibrium admits a spectral characterization: the encoder transmits information only along eigenspaces associated with the negative eigenvalues of a mismatch matrix. In the noisy regime, we derive analytical thresholds for informative signaling, showing that communication collapses if the sensitivity mismatch or transmission cost exceeds a channel-dependent threshold.
\end{abstract}

\begin{keyword}
Signaling games, Cheap talk, Bayesian persuasion, Linear bias.
\end{keyword}
\end{frontmatter}
\section{Introduction}
\label{sec:intro}
In decentralized control and cyber-physical systems, information exchange often involves agents with misaligned incentives. Unlike classical communication, which focuses on reliable data reconstruction, \emph{Strategic Information Transmission} (SIT) arises when an informed sender communicates with a receiver whose actions impact both parties' costs. In this setting, divergent objectives create a trade-off between information revelation and strategic manipulation. Characterizing the limits of communication under such misalignment is essential for designing robust decentralized systems, with applications ranging from smart grids \cite{larrousse2014crawford-sobel} and federated learning \citep{sun2024strategic,munif2024strategic} to economic and political interactions \cite{kamenica2019bayesian}.

In the foundational cheap-talk model of \cite{crawford1982strategic}, costless but non-verifiable communication between misaligned agents leads to coarse, quantized information transmission. The Bayesian persuasion framework of \cite{kamenica2011Bayesian} considers a complementary setting where the sender commits to an information disclosure policy, yielding a Stackelberg game structure that enables more expressive communication.

Over the past decade, these ideas have been extended to control, signal processing, and information theory. In these domains, the sender is typically a sensor or encoder, the receiver is an estimator or controller, and both agents have quadratic objectives with Gaussian information structures \cite{akyol2017information-theoretic,saritas2017quadratic,farokhi2017estimation}. This has produced an extensive literature on signaling games, including conditions for linear equilibria (a problem dating back to \cite{witsenhausen1968acounterexample}), dynamic extensions \cite{saritas2020dynamic,sayin2019hierarchical}, privacy-signaling games \cite{akyol2015privacy,kazikli2022quadratic}, and multidimensional geometric characterizations \cite{kazikli2023signaling}. In parallel, the information-design literature has developed semidefinite programming methods and spectral characterizations for quadratic Gaussian persuasion \cite{tamura2018bayesian,sayin2022bayesian,velicheti2025value}, establishing that senders transmit information along state-space directions that serve their objectives and block directions that do not.

Most of this literature models preference misalignment as an \emph{additive bias}, assuming the encoder wants the decoder's action to track a shifted target instead of the true state. Under additive bias, Bayesian persuasion and Stackelberg cheap talk models generally yield fully revealing equilibria. Because the encoder anticipates the decoder's rational response, withholding information provides no strategic advantage \cite{saritas2017quadratic,saritas2020dynamic}. However, misalignment in modern cyber-physical and human-machine systems frequently stems from a difference in \emph{sensitivity} rather than a simple shift in location. Two agents might monitor the same vector-valued state but apply different linear transformations to it. For example, safety protocols or privacy constraints \cite{akyol2015privacy} might make a strategic sensor overly sensitive to specific dimensions. Consequently, the encoder actually prefers the decoder's action to track a scaled or rotated version of the state, not a shifted one.

This paper investigates how such \emph{linear sensitivity mismatch} alters strategic Gaussian signaling. We consider a sender who observes a Gaussian source and communicates over either a noiseless (cheap-talk) or noisy channel to a receiver. The receiver minimizes standard mean-squared estimation error (MMSE), while the sender's objective is centered at a linear transformation of the state. The sender announces her policy first, and the receiver best responds. This sequential structure forms a Stackelberg game. Our main contributions are:

\begin{enumerate}[wide]
    \item We introduce a Gaussian signaling game where misalignment is captured by a general matrix $A$ that transforms the state. This formulation generalizes the standard additive bias models. We show that equilibrium behavior depends on the eigenvalues of the sensitivity mismatch rather than the magnitude of a bias vector.
    
    \item For multidimensional cheap talk, we specialize quadratic Gaussian persuasion results \citep{tamura2018bayesian,sayin2022bayesian} to the linear sensitivity mismatch model. We establish that equilibrium disclosure is governed by the spectral decomposition of a transformed \emph{mismatch} matrix, and we provide an explicit linear encoder that realizes the optimal posterior covariance.
 
    \item Unlike persuasion literature, we analyze the \emph{noisy} channel setting with transmission costs. We derive necessary and sufficient conditions for the existence of informative equilibria in the scalar case, alongside a necessary condition in the vector case with isotropic sensitivity. We show a phase transition where communication collapses if the transmission cost exceeds a threshold determined by the channel capacity and the sensitivity mismatch.
\end{enumerate}

\textbf{Notation.} We denote vectors with bold lower-case letters (e.g., $\boldsymbol{x}$) and matrices with regular uppercase letters (e.g., $A,\Sigma$). Scalars are denoted by lower-case non-bold letters (e.g., $\rho,a$ ). Unless specified otherwise, random variables and their realizations share the same symbols, distinguished by context. The transpose is $(\cdot)^\top$, and $\tr(\cdot)$ is the matrix trace. $I$ and $O$ denote the identity and zero matrices, respectively. For two symmetric matrices $A$ and $B$, $A \succ B$ ($A \succeq B$) indicates that $A-B$ is positive (semi)definite. $\mathcal{N}(\boldsymbol{0}, \Sigma)$ denotes a Gaussian distribution with mean $\boldsymbol{0}$ and covariance $\Sigma$. $\mathbb{E}[\cdot]$ denotes the expectation operator, $\|\cdot\|$ denotes the standard Euclidean norm.

\section{Problem Formulation}
\label{sec:problem}
\subsection{System Model}
We consider a strategic communication system consisting of two decision-makers: an \emph{encoder} (sender), and a \emph{decoder} (receiver). The system operates as follows (see Fig.~\ref{fig:system}):
\begin{enumerate}[wide,label=\textbf{\arabic*.},itemindent=10pt,topsep=-\baselineskip]
    \item \emph{Source:} A random vector $\boldsymbol{m} \in \mathbb{R}^n$ is drawn from a zero-mean Gaussian distribution with covariance $\Sigma_m \succ O$, i.e., $\boldsymbol{m} \sim \mathcal{N}(\boldsymbol{0}, \Sigma_m)$.
    
    \item \emph{Encoding:} The encoder observes the realization $\boldsymbol{m}$ of the source and transmits a signal $\boldsymbol{x} = \gamma^{\mathrm{e}}(\boldsymbol{m}) \in \mathbb{R}^n$. We define the set of admissible encoder policies, $\Gamma^{\mathrm{e}}$, as the set of all deterministic (Borel-measurable) functions $\gamma^{\mathrm{e}}: \mathbb{R}^n \to \mathbb{R}^n$.
    
    \item \emph{Channel:} The signal $\boldsymbol{x}$ passes through an additive Gaussian noise channel $\boldsymbol{y} = \boldsymbol{x} + \boldsymbol{w}$, where $\boldsymbol{w} \sim \mathcal{N}(\boldsymbol{0}, \Sigma_w)$ is independent of $\boldsymbol{m}$. The noiseless (cheap-talk) setting corresponds to $\Sigma_w = O$, yielding $\boldsymbol{y} = \boldsymbol{x}$.
    
    \item \emph{Decoding:} The decoder observes the realization $\boldsymbol{y}$ and produces an estimate $\boldsymbol{u} = \gamma^{\mathrm{d}}(\boldsymbol{y}) \in \mathbb{R}^n$. Similarly, we define the set of admissible decoder policies, $\Gamma^{\mathrm{d}}$, as the set of all deterministic (Borel-measurable) functions $\gamma^{\mathrm{d}}: \mathbb{R}^n \to \mathbb{R}^n$.
\end{enumerate}

\begin{figure}[htbp]
    \centering
    \begin{tikzpicture}[
    >=Latex, 
    node distance=1.3cm, 
    thick,
    block/.style={draw, rectangle, minimum height=3em, minimum width=4.5em, align=center, rounded corners},
    encstyle/.style={block, fill=sendercolor!10, draw=sendercolor, text=sendercolor},
    decstyle/.style={block, fill=receivercolor!10, draw=receivercolor, text=receivercolor},
    sum/.style={draw, circle, inner sep=0pt, minimum size=4mm},
    font=\small
]
    \node (source) {$\boldsymbol{m}$};
    \node[encstyle, right=0.6cm of source] (encoder) {\textbf{Encoder} \\ $\gamma^e$};
    \node[sum, right=1.1cm of encoder] (sum) {+};
    \node[above=0.8cm of sum] (noise) {$\boldsymbol{w} \sim \mathcal{N}(\boldsymbol{0}, \Sigma_w)$};
    \node[decstyle, right=1.1cm of sum] (decoder) {\textbf{Decoder} \\ $\gamma^d$};
    \node[right=0.6cm of decoder] (dest) {$\boldsymbol{u}$};
    \draw[->] (source) -- (encoder);
    \draw[->] (encoder) -- node[above] {$\boldsymbol{x}$} (sum);
    \draw[->] (noise) -- (sum);
    \draw[->] (sum) -- node[above] {$\boldsymbol{y}$} (decoder);
    \draw[->] (decoder) -- (dest);
\end{tikzpicture}
    \caption{System model. The encoder observes the source $\boldsymbol{m}$, transmits $\boldsymbol{x}$ over a noisy channel, and the decoder produces an estimate $\boldsymbol{u}$ of $\boldsymbol{m}$.}
    \label{fig:system}
\end{figure}
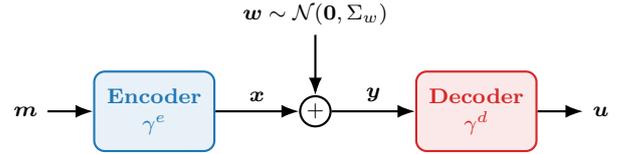
\subsection{Objective Functions with Linear Bias}
The encoder and decoder aim to minimize their respective expected costs. These objectives are \emph{non-aligned}, making the problem a game rather than a joint optimization problem. We define the \emph{instantaneous cost functions} for the decoder and encoder as follows:
\begin{align}
    c^{\mathrm{d}}(\boldsymbol{m}, \boldsymbol{u}) &\coloneqq \|\boldsymbol{m} - \boldsymbol{u}\|^2, \label{eq:inst_cost_d} \\
    c^{\mathrm{e}}(\boldsymbol{m}, \boldsymbol{x}, \boldsymbol{u}) &\coloneqq \|A\boldsymbol{m} - \boldsymbol{b} - \boldsymbol{u}\|^2 + \rho \|\boldsymbol{x}\|^2. \label{eq:inst_cost_e}
\end{align}
Here, $c^{\mathrm{d}}$ denotes the standard squared-error distortion. For the encoder, $c^{\mathrm{e}}$ contains a \emph{linear bias} through a matrix $A \in \mathbb{R}^{n \times n}$, an \emph{additive bias} $\boldsymbol{b} \in \mathbb{R}^n$, and a soft transmission power constraint weighted by $\rho \geq 0$.\\
Each decision maker minimizes the expectation of its cost under the joint distribution induced by the source, channel noise, and the chosen policies. Expected costs are denoted as $J^{\mathrm{d}}(\gamma^{\mathrm{e}}, \gamma^{\mathrm{d}}) \coloneqq \mathbb{E}[c^{\mathrm{d}}(\boldsymbol{m}, \boldsymbol{u})]$ and $J^{\mathrm{e}}(\gamma^{\mathrm{e}}, \gamma^{\mathrm{d}}) \coloneqq \mathbb{E}[c^{\mathrm{e}}(\boldsymbol{m}, \boldsymbol{x}, \boldsymbol{u})]$ for the decoder and encoder, respectively.

\begin{rem}
The matrix $A$ captures the \emph{sensitivity mismatch} between the encoder and decoder. If $A = I$, this setup recovers the classical model of \cite{saritas2017quadratic} with additive bias $\boldsymbol{b}$.
\end{rem}
\begin{rem}[Cheap talk vs.\ Signaling]
When $\rho = 0$ and the channel is noiseless, the transmitted signal $\boldsymbol{x}$ has no direct cost and serves purely as a message. This setting is referred to as Stackelberg \emph{cheap talk} \citep{crawford1982strategic} (with commitment), effectively making the problem one of \emph{Bayesian Persuasion} \citep{kamenica2011Bayesian}. When $\rho > 0$ or the channel is noisy, the problem becomes a \emph{signaling game}, where the encoder faces a trade-off between communication fidelity and transmission cost.
\end{rem}
\subsection{Stackelberg Equilibrium}
We study the \emph{Stackelberg game} where the encoder acts as the \emph{leader} and the decoder as the \emph{follower}. This models scenarios where the encoder's policy is designed and committed to before the decoder optimizes its response.\\
\begin{defn}[Stackelberg Equilibrium]
A pair of policies constitutes a Stackelberg Equilibrium (SE) with the encoder as the leader if the encoder commits to a policy $\gamma^{\mathrm{e}}$ and the decoder plays the optimal best response. The policy $\gamma^{*,e}$ satisfies
\begin{equation*}
    J^{\mathrm{e}}(\gamma^{*,e}, \gamma^{*,d}(\gamma^{*,e})) \leq J^{\mathrm{e}}(\gamma^{\mathrm{e}}, \gamma^{*,d}(\gamma^{\mathrm{e}})) \quad \forall \gamma^{\mathrm{e}} \in \Gamma^{\mathrm{e}},
\end{equation*}
where $\gamma^{*,d}(\gamma^{\mathrm{e}})$ is the decoder's best response to $\gamma^{\mathrm{e}}$:
\begin{equation*}
    \gamma^{*,d}(\gamma^{\mathrm{e}})= \arg\min_{\gamma^{\mathrm{d}} \in \Gamma^{\mathrm{d}}} J^{\mathrm{d}}(\gamma^{\mathrm{e}}, \gamma^{\mathrm{d}}).
\end{equation*}
\end{defn}

\begin{prop}[Decoder's best response]
\label{prop:decoder_mmse}
For any fixed encoder policy $\gamma^{\mathrm{e}}$, the unique optimal decoder strategy is the conditional expectation of the source given the observation: $\gamma^{*,d}(\boldsymbol{y}) = \mathbb{E}[\boldsymbol{m} \mid \boldsymbol{y}]$.
\end{prop}

\begin{pf}
Let $\widehat{\boldsymbol{m}}(\boldsymbol{y}) \coloneqq\mathbb{E}[\boldsymbol{m} \mid \boldsymbol{y}]$ and let $\boldsymbol{u}(\boldsymbol{y})$ be any arbitrary estimator. We decompose the mean squared error:
\begin{align*}
    \mathbb{E}[\|\boldsymbol{m} - \boldsymbol{u}\|^2] &= \mathbb{E}[\|(\boldsymbol{m} - \widehat{\boldsymbol{m}}) + (\widehat{\boldsymbol{m}} - \boldsymbol{u})\|^2] \\
    &= \mathbb{E}[\|\boldsymbol{m} - \widehat{\boldsymbol{m}}\|^2] + \mathbb{E}[\|\widehat{\boldsymbol{m}} - \boldsymbol{u}\|^2] \\
    &\quad + 2\mathbb{E}[(\widehat{\boldsymbol{m}} - \boldsymbol{u})^\top (\boldsymbol{m} - \widehat{\boldsymbol{m}})].
\end{align*}
By the law of iterated expectations, we have
\begin{align*}
    \mathbb{E}\big[(\widehat{\boldsymbol{m}} - \boldsymbol{u})^\top (\boldsymbol{m}-\widehat{\boldsymbol{m}})\big]
    &= \mathbb{E}\Big[\mathbb{E}[(\widehat{\boldsymbol{m}} - \boldsymbol{u})^\top(\boldsymbol{m}-\widehat{\boldsymbol{m}})\mid \boldsymbol{y}]\Big]=0.
\end{align*}
Thus, the cost is minimized if and only if $\mathbb{E}[\|\widehat{\boldsymbol{m}} - \boldsymbol{u}\|^2] = 0$, which implies $\boldsymbol{u}(\boldsymbol{y}) = \widehat{\boldsymbol{m}}(\boldsymbol{y})$ almost surely.
\hfill $\qed$
\end{pf}

We define the following equilibrium outcomes.
\begin{defn}
    We say that an SE is \emph{informative} if the encoder reveals information related to the source, i.e., the source $\boldsymbol{m}$ and the message $\boldsymbol{x}$ are not independent random variables. An SE is \emph{non-informative} if the signal provided by the encoder does not alter the decoder's belief relative to the prior mean, resulting in $\boldsymbol{u} = \mathbb{E}[\boldsymbol{m}]$ almost surely. Furthermore, an SE is \emph{fully revealing} if the signal $\boldsymbol{x}$ allows the decoder to perfectly reconstruct the source (i.e., $\boldsymbol{u}=\boldsymbol{m}$ almost surely).
\end{defn}

\section{Noiseless Case: Bayesian Persuasion}
\subsection{Scalar Case}
\label{sec:cheaptalk_scalar}

We begin by analyzing the scalar instance ($n=1$) under the cheap-talk assumption with no noise ($w=0,\rho = 0$). In this setting, the signal is $y = x$. By Prop.~\ref{prop:decoder_mmse}, for any encoding policy $\gamma^{\mathrm{e}}$, the decoder's unique optimal response is $u = \mathbb{E}[m \mid x]$.

The encoder, anticipating this response, seeks a policy $\gamma^{\mathrm{e}}$ that minimizes $\mathbb{E}[(am - b - u)^2]$.
\begin{thm}[Stackelberg informative threshold]
The nature of the SE is determined solely by the multiplicative bias parameter $a$:
\begin{enumerate}[label=\textbf{\arabic*.},topsep=-0.5\baselineskip]
    \item If $a > 1/2$, every SE is fully revealing.
    \item If $a < 1/2$, every SE is non-informative (babbling).
    \item If $a = 1/2$, the encoder is indifferent and every policy is an SE.
\end{enumerate}
\end{thm}
\begin{pf}
The game proceeds sequentially. Given the decoder's optimal response $u = \mathbb{E}[m \mid x]$ (Prop.~\ref{prop:decoder_mmse}), the encoder selects $\gamma^{\mathrm{e}}$ to minimize $
 J^{\mathrm{e}} = \mathbb{E} \left[ (am - b - u)^2 \right]$.
Expanding the quadratic term: $J^{\mathrm{e}}= \mathbb{E} \left[ (am - u)^2 \right] - 2b\mathbb{E}[am - u] + b^2$.
We analyze the linear and quadratic terms separately. By the law of iterated expectations, $\mathbb{E}[u] = \mathbb{E}[\mathbb{E}[m|x]] = \mathbb{E}[m]$. Thus, the linear term becomes $-2b(a\mathbb{E}[m] - \mathbb{E}[u]) = -2b(a-1)\mathbb{E}[m]$.
This term depends only on the prior statistics of the source and the constants $a$ and $b$; it is independent of the encoder's policy $\gamma^{\mathrm{e}}$.

For the quadratic term, we decompose $am - u = (a-1)u + a(m - u)$. By the orthogonality principle, the error $m - u$ is orthogonal to the estimator $u$. Thus, $\mathbb{E}[u(m - u)] = 0$, and variance decomposition gives $\mathbb{E}[(m - u)^2] = \mathbb{E}[m^2] - \mathbb{E}[u^2]$. Expanding the expectation yields:
\begin{align*}
\mathbb{E}[ (am - u)^2 ] &= (a-1)^2\mathbb{E}[u^2] + a^2\mathbb{E}[(m - u)^2] \\
&= \left((a-1)^2 - a^2\right)\mathbb{E}[u^2] + a^2\mathbb{E}[m^2] \\
&= (1-2a)\mathbb{E}[u^2] + a^2\mathbb{E}[m^2].
\end{align*}
The total cost for the encoder is therefore $J^{\mathrm{e}} = (1-2a)\mathbb{E}[u^2] + \mathrm{cte}$, where $\mathrm{cte} = a^2\mathbb{E}[m^2] - 2b(a-1)\mathbb{E}[m] + b^2$ is a constant independent of $\gamma^{\mathrm{e}}$. The minimization problem reduces to minimizing $(1-2a)\mathbb{E}[u^2]$. To minimize the cost, the encoder must choose $\mathbb{E}[u^2]$ subject to $(\mathbb{E}[m])^2 \leq \mathbb{E}[u^2] \leq \mathbb{E}[m^2]$. We have the following cases:
\begin{enumerate}[wide,topsep=-1\baselineskip,label=\textbf{\arabic*.}]
\item If $a > 1/2$, the coefficient $(1-2a)$ is negative. The encoder minimizes cost by maximizing $\mathbb{E}[u^2]$. The upper bound $\mathbb{E}[u^2] = \mathbb{E}[m^2]$ is achievable by any injective policy, such as the identity map $\gamma^{\mathrm{e}}(m) = m$. In this case, $u = m$ almost surely, resulting in a fully revealing equilibrium.
\item If $a < 1/2$, then $(1-2a)>0$. The encoder minimizes cost by minimizing $\mathbb{E}[u^2]$. The lower bound $\mathbb{E}[u^2] = (\mathbb{E}[m])^2$ is achievable by any constant policy, such as $\gamma^{\mathrm{e}}(m) = 0$. In this case, the signal carries no information ($u = \mathbb{E}[m]$). The equilibrium is non-informative.
\item If $a = 1/2$, the coefficient is zero. The cost is independent of the policy $\gamma^{\mathrm{e}}$, meaning any admissible policy constitutes a Stackelberg equilibrium. \hfill $\qed$
\end{enumerate}
\end{pf}
\begin{rem}
The $a=1/2$ threshold reflects the point where the benefit of inducing variance in the decoder's estimate perfectly offsets the estimation error penalty. This boundary relies on the decoder's quadratic loss; minimizing a non-quadratic metric (e.g., absolute error) would alter the optimal response, breaking the orthogonality condition $\mathbb{E}[u(m-u)]=0$ and shifting the threshold.
\end{rem}

\subsection{Multidimensional Case}
\label{sec:cheaptalk_multidim}
We now consider the multidimensional cheap-talk game described in Section~\ref{sec:problem}, where the source is $\boldsymbol{m} \sim \mathcal{N}(\boldsymbol{0}, \Sigma_m)$ with $\Sigma_m \succ O$, and the encoder has a bias matrix $A$. In this cheap-talk setting, we have no noise ($\Sigma_w = O$), and signaling is cost-free ($\rho = 0$).\\
For any fixed encoder policy $\gamma^{\mathrm{e}}$, the decoder observes $\boldsymbol{x} = \gamma^{\mathrm{e}}(\boldsymbol{m})$. The decoder's objective is to minimize the mean squared error; the unique optimal best response (by Prop.~\ref{prop:decoder_mmse}): $\gamma^{*,d}(\boldsymbol{x}) = \mathbb{E}[\boldsymbol{m} \mid \boldsymbol{x} = \gamma^{\mathrm{e}}(\boldsymbol{m})]$.

Let $\boldsymbol{u} = \gamma^{*,d}(\boldsymbol{x})$ denote the decoder's estimate. By the properties of conditional expectation, we can decompose the source into the estimate and the error $\boldsymbol{m} - \boldsymbol{u}$. The error is orthogonal to the estimate, i.e., $\mathbb{E}[\boldsymbol{u}(\boldsymbol{m} - \boldsymbol{u})^\top] = O$. Consequently, the covariance of the source decomposes as $\Sigma_m = \Sigma_u + \Sigma_{e}$, where $\Sigma_u = \mathbb{E}[\boldsymbol{u}\boldsymbol{u}^\top]$ and $\Sigma_{e} = \mathbb{E}[(\boldsymbol{m} - \boldsymbol{u})(\boldsymbol{m} - \boldsymbol{u})^\top]$. Since covariance matrices are positive semidefinite, any achievable posterior mean covariance $\Sigma_u$ must satisfy the constraint $\Sigma_m \succeq \Sigma_u \succeq O$.

Expanding the norm and noting that $\mathbb{E}[\boldsymbol{m}] = \mathbb{E}[\boldsymbol{u}]$ (by the law of iterated expectations), the linear terms involving $\boldsymbol{b}$ reduce to constants independent of the policy. The encoder's problem reduces to minimizing the quadratic term: $\mathbb{E}[\|A\boldsymbol{m} - \boldsymbol{u}\|^2] = \mathbb{E}[\|(A-I)\boldsymbol{u} + A(\boldsymbol{m} - \boldsymbol{u})\|^2].$

Using the orthogonality $\mathbb{E}[\boldsymbol{u}(\boldsymbol{m} - \boldsymbol{u})^\top] = O$ and introducing the trace operator (recall that $\mathbb{E}[\|\boldsymbol{z}\|^2] = \tr(\mathbb{E}[\boldsymbol{z}\boldsymbol{z}^\top])$), we rewrite the cost as
\begin{align*}
     &\tr\left((A-I)^\top(A-I) \Sigma_u\right) + \tr\left(A^\top A \Sigma_{e}\right) \nonumber \\
      &= \tr\left((A-I)^\top(A-I) \Sigma_u\right) + \tr\left(A^\top A (\Sigma_m - \Sigma_u)\right).
\end{align*}
Grouping terms dependent on $\Sigma_u$, the encoder's objective becomes:
\begin{equation}
\label{eq:sender_problem_ct_multidim}
    J^{\mathrm{e}}(\Sigma_u) = \tr(V \Sigma_u) + \mathrm{const},
\end{equation}
where $\mathrm{const}=\tr\left(A^\top A \Sigma_m\right)$, and $V \coloneq (A-I)^\top(A-I) - A^\top A = I - (A + A^\top)$ is the cost kernel. The optimization is over the set of achievable posterior mean covariances $\Sigma_u$, which must satisfy $\Sigma_m \succeq \Sigma_u \succeq O$.\\
Directly optimizing over $\Sigma_u$ is difficult due to the generalized inequality constraint $\Sigma_m \succeq \Sigma_u$. To solve this, we utilize a transformation to standardize the constraint, adapted from \cite{velicheti2023strategic}.

\begin{lem}[Equivalent SDP formulation]
    \label{lem:sdp_equivalence}
Let $\Sigma_m \succ O$. The optimization problem: $\min_{\Sigma_u} \tr(V \Sigma_u)$ s.t. $\Sigma_m \succeq \Sigma_u \succeq O$, can be equivalently written as:
\begin{equation}
\label{eq:transformed_prob}
    \min_{\Pi \in \mathbb{S}^n} \quad \tr(B \Pi) \quad \text{s.t.} \quad I \succeq \Pi \succeq O,
\end{equation}
where $B = \Sigma_m^{\frac{1}{2}} V \Sigma_m^{\frac{1}{2}}$, $\mathbb{S}^n$ denotes the set of positive semi-definite matrices of dimension $n\times n$, and the optimization variable is transformed via $\Sigma_u = \Sigma_m^{\frac{1}{2}} \Pi \Sigma_m^{\frac{1}{2}}$.
\end{lem}

\begin{pf}
Since $\Sigma_m \succ O$, the matrix $\Sigma_m^{\frac{1}{2}}$ exists and is invertible. We introduce the change of variable $\Pi = \Sigma_m^{-\frac{1}{2}} \Sigma_u \Sigma_m^{-\frac{1}{2}}$. First, we transform the objective function using the definition $\Sigma_u = \Sigma_m^{\frac{1}{2}} \Pi \Sigma_m^{\frac{1}{2}}$ and the cyclic property of the trace operator\footnote{i.e., $\tr(XYZ) = \tr(ZXY)$}:
\begin{align*}
    \tr(V \Sigma_u) &= \tr(V \Sigma_m^{\frac{1}{2}} \Pi \Sigma_m^{\frac{1}{2}}) \\
    &= \tr(\Sigma_m^{\frac{1}{2}} V \Sigma_m^{\frac{1}{2}} \Pi) = \tr(B \Pi).
\end{align*}
Second, we transform the constraints. Recall that for any invertible matrix $M$, $A \succeq B$ if and only if $M A M^\top \succeq M B M^\top$. We apply this with $M = \Sigma_m^{-\frac{1}{2}}$:
\begin{align*}
    \Sigma_m \succeq \Sigma_u \succeq O 
    &\iff \Sigma_m^{-\frac{1}{2}} \Sigma_m \Sigma_m^{-\frac{1}{2}} \succeq \Sigma_m^{-\frac{1}{2}} \Sigma_u \Sigma_m^{-\frac{1}{2}} \\
    &\qquad\qquad \succeq \Sigma_m^{-\frac{1}{2}} O \Sigma_m^{-\frac{1}{2}} \\
    &\iff I \succeq \Pi \succeq O.
\end{align*}
Thus, the problem is equivalent to minimizing $\tr(B \Pi)$ subject to $I \succeq \Pi \succeq O$.
\hfill $\qed$
\end{pf}

\begin{rem}[Relation to quadratic persuasion]
The sender's problem in \eqref{eq:sender_problem_ct_multidim} coincides with the quadratic persuasion formulations in \citep{tamura2018bayesian,sayin2022bayesian}, where the sender's cost is linear in $\Sigma_u$ and the feasible set is $\{0 \preceq \Sigma_u \preceq \Sigma_m\}$. In our setting, the linear functional $V$ is induced by the sensitivity mismatch matrix $A$, and Prop.~\ref{prop:SE_policy} below shows how to construct an explicit linear encoder achieving the optimal $\Sigma_u$, thereby instantiating the persuasion solution in a Gaussian SIT context.
\end{rem}
Applying Lemma~\ref{lem:sdp_equivalence} with the sender's cost kernel $V = I - (A+A^\top)$, we define the weight matrix $B = \Sigma_m^{\frac{1}{2}} V \Sigma_m^{\frac{1}{2}}$. The solution to the transformed problem \eqref{eq:transformed_prob} allows us to characterize the SE as follows.
\begin{thm}[SE information structure]
    \label{thm:SE_covariance}
    Let $\beta_1 \leq \beta_2 \leq \dots \leq \beta_n$ be the ordered eigenvalues of $B = \Sigma_m^{\frac{1}{2}} V \Sigma_m^{\frac{1}{2}}$, and let $\boldsymbol{q}_1, \dots, \boldsymbol{q}_n$ be the corresponding orthonormal eigenvectors. Let $k\in\mathbb{N}$ denote the number of strictly negative eigenvalues (i.e., $\beta_k < 0$ and $\beta_{k+1} \geq 0$).\\
    Among the set of optimal policies, let us select the solution minimizing the rank of $\Sigma_u$.\footnote{If any eigenvalues are exactly zero, the equilibrium is not unique. Similar to \cite{velicheti2023strategic,sayin2022bayesian}, we effectively select the solution with the minimum rank (least informative) among the set of optimal policies.}The resulting unique SE posterior mean covariance is given by
    \begin{equation}
        \Sigma_u^* = \Sigma_m^{\frac{1}{2}} \Pi^* \Sigma_m^{\frac{1}{2}},
    \end{equation}
    where $\Pi^* = \sum_{i=1}^k \boldsymbol{q}_i \boldsymbol{q}_i^\top$ is the projection onto the subspace spanned by the eigenvectors associated with the strictly negative eigenvalues of $B$. In particular:
\begin{enumerate}[label=\textbf{\arabic*.},topsep=-0.5\baselineskip]
        \item If $k=0$ (all $\beta_i \geq 0$), the equilibrium is non-informative ($\Sigma_u^* = O$).
        \item If $k=n$ (all $\beta_i < 0$), the equilibrium is fully revealing ($\Sigma_u^* = \Sigma_m$).
        \item If $0 < k < n$, the equilibrium is partially revealing.
\end{enumerate}
\end{thm}
\begin{pf}
Following Lemma~\ref{lem:sdp_equivalence}, we minimize $\tr(B\Pi)$ subject to
$O \preceq \Pi \preceq I$. We perform the spectral decomposition
$B = Q\Lambda Q^\top$, where $\Lambda = \operatorname{diag}(\beta_1, \dots, \beta_n)$
and $Q = [\boldsymbol{q}_1 \dots \boldsymbol{q}_n]$. Let $F = Q^\top \Pi Q$.
Using the cyclic property of the trace, we obtain $\tr(B\Pi) = \tr(Q \Lambda Q^\top \Pi)
    = \tr(\Lambda Q^\top \Pi Q)
    = \tr(\Lambda F)$.
The constraint $O \preceq \Pi \preceq I$ is equivalent to $O \preceq F \preceq I$
since $Q$ is orthogonal. Thus, we are solving $\min_{F\in\mathbb{S}^n} \ \tr(\Lambda F)\quad \text{s.t.} \quad O \preceq F \preceq I$.

This is the same SDP (up to a sign change in the cost matrix) as in
\cite[Thm~1]{tamura2018bayesian}, where it is shown that there exists an
optimal solution to the SDP that is an orthogonal projection matrix in the eigenbasis
of the cost matrix. Thus, without loss of generality, we may restrict attention to $F$ that
are diagonal in the eigenbasis of $B$, with diagonal entries in $\{0,1\}$. Hence we can write $F = \operatorname{diag}(f_{11},\dots,f_{nn})$ with
$f_{ii} \in \{0,1\}$, and the objective becomes the following sum $\tr(\Lambda F) = \sum_{i=1}^n \beta_i f_{ii}$.
To minimize it, each $f_{ii}$ is chosen based on the sign of $\beta_i$:
\begin{equation*}
    f_{ii}^* =
    \begin{cases}
        1 & \text{if } \beta_i < 0, \\
        0 & \text{if } \beta_i \geq 0.
    \end{cases}
\end{equation*}
(For $\beta_i = 0$, the choice does not affect the cost; we set $f_{ii} = 0$
to obtain the minimum-rank, least-informative solution.)

Since the eigenvalues are sorted such that the first $k$ are negative, the
optimal matrix $F^*$ is diagonal with the first $k$ entries equal to $1$ and
the rest $0$. Transforming back yields
\begin{equation*}
    \Pi^* = Q F^* Q^\top
          = Q
            \begin{bmatrix}
                I_k & O \\
                O   & O
            \end{bmatrix}
            Q^\top
          = Q_k Q_k^\top,
\end{equation*}
where $Q_k \in \mathbb{R}^{n \times k}$ is the matrix containing the first
$k$ columns of $Q$ (the eigenvectors corresponding to negative eigenvalues).
\hfill $\qed$
\end{pf}
The following proposition, adapted from \cite[Thm~2]{tamura2018bayesian}, characterizes an encoding policy that achieves the equilibrium described in Theorem~\ref{thm:SE_covariance}.

\begin{prop}[Equilibrium signaling policy]
    \label{prop:SE_policy}
    Let $k$ be the number of strictly negative eigenvalues of $B$, and let $Q_k \in \mathbb{R}^{n \times k}$ be the matrix of the corresponding eigenvectors (the first $k$ columns of $Q$). 
    
    The equilibrium covariance $\Sigma_u^*$ characterized in Theorem~\ref{thm:SE_covariance} is achieved by a deterministic linear encoder policy $\gamma^{*,e}(\boldsymbol{m}) = L \boldsymbol{m}$, where $L \in \mathbb{R}^{n \times n}$ is constructed as
    \begin{equation}
        L = \begin{bmatrix} Q_k^\top \Sigma_m^{-\frac{1}{2}} \\ O_{(n-k) \times n} \end{bmatrix},
    \end{equation}
    where $O_{(n-k) \times n}$ is a zero matrix padding the remaining dimensions.
\end{prop}

\begin{pf}
The decoder's best response to a linear Gaussian map $\boldsymbol{x} = L\boldsymbol{m}$ is the linear MMSE estimator, explicitly given by $\boldsymbol{u} = \Sigma_{mx}\Sigma_{xx}^{\dagger}\boldsymbol{x}$, where $(\cdot)^\dagger$ denotes the pseudoinverse of a matrix $\Sigma_{xx}$, which handles the rank deficiency of the noiseless signal covariance, ensuring the estimator is well-defined on the signal subspace.\\
First, we compute the covariance of the signal $\boldsymbol{x}$:
\begin{align*}
\Sigma_{xx} &= L \Sigma_m L^\top \\
&= \begin{bmatrix} Q_k^\top \Sigma_m^{-\frac{1}{2}} \\ O \end{bmatrix} \Sigma_m \begin{bmatrix} \Sigma_m^{-\frac{1}{2}} Q_k & O^\top \end{bmatrix} \\
&= \begin{bmatrix} Q_k^\top Q_k & O \\ O & O \end{bmatrix} = \begin{bmatrix} I_k & O \\ O & O_{n-k} \end{bmatrix},
\end{align*}
where we used the orthonormality condition $Q_k^\top Q_k = I_k$.
Next, the cross-covariance $\Sigma_{mx}$ is
\[
\Sigma_{mx} = \Sigma_m L^\top = \Sigma_m \begin{bmatrix} \Sigma_m^{-\frac{1}{2}} Q_k & O^\top \end{bmatrix} = \begin{bmatrix} \Sigma_m^{\frac{1}{2}} Q_k & O \end{bmatrix}.
\]
The posterior estimate is $\boldsymbol{u} = \Sigma_{mx} \Sigma_{xx}^{\dagger} \boldsymbol{x}$. Noting that $\Sigma_{xx}^{\dagger} = \Sigma_{xx}$ (as it is a diagonal projection matrix), the covariance of the estimate is
\begin{align*}
\Sigma_u &= \Sigma_{mx} \Sigma_{xx}^{\dagger} \Sigma_{mx}^\top \\
&= \begin{bmatrix} \Sigma_m^{\frac{1}{2}} Q_k & O \end{bmatrix} \begin{bmatrix} I_k & O \\ O & O \end{bmatrix} \begin{bmatrix} Q_k^\top \Sigma_m^{\frac{1}{2}} \\ O^\top \end{bmatrix} \\
&= \Sigma_m^{\frac{1}{2}} Q_k I_k Q_k^\top \Sigma_m^{\frac{1}{2}} 
= \Sigma_m^{\frac{1}{2}} (Q_k Q_k^\top) \Sigma_m^{\frac{1}{2}}.
\end{align*}
Since $\Pi^* = Q_k Q_k^\top$ (from the proof of Theorem~\ref{thm:SE_covariance}), we have $\Sigma_u = \Sigma_m^{\frac{1}{2}} \Pi^* \Sigma_m^{\frac{1}{2}} = \Sigma_u^*$. Thus, the policy achieves the SE.
\hfill $\qed$
\end{pf}

\section{Noisy Case: Signaling Game}
\label{sec:noisy}
We now turn to the general setting where the channel is noisy and the encoder faces a transmission cost ($\rho > 0$). Unlike the cheap-talk setting, the signal is corrupted by noise, preventing the decoder from perfectly inverting the encoder's map even if it is injective.

\subsection{Scalar Case}
\label{sec:signaling_scalar}
We first analyze the scalar case ($n=1$) where $m \sim \mathcal{N}(0, \sigma_m^2)$ and $w \sim \mathcal{N}(0, \sigma_w^2)$ with $\sigma_w^2 > 0$. The encoder chooses a policy $\gamma^{\mathrm{e}}: \mathbb{R} \to \mathbb{R}$, resulting in $x = \gamma^{\mathrm{e}}(m)$ and $y = x + w$.

\begin{thm}[Stackelberg signaling thresholds]
    \label{thm:noisy_scalar}
    \setlength{\abovedisplayskip}{3pt}
    \setlength{\belowdisplayskip}{3pt}
    \setlength{\abovedisplayshortskip}{0pt}
    \setlength{\belowdisplayshortskip}{0pt}
    
    The existence and structure of SE are determined by the bias parameter $a$:
\begin{enumerate}[label=\textbf{\arabic*.},topsep=-\baselineskip] 
\item If $a \leq \tfrac{1}{2}$, the unique SE is non-informative ($P^*=0$).
\item If $a >\tfrac{1}{2}$, an informative SE exists if and only if the transmission cost satisfies
\begin{equation}
    \label{eq:rho_condition}
    0 < \rho < \frac{\sigma_m^2}{\sigma_w^2} (2a - 1).
\end{equation}
where the optimal transmission power is given by
\begin{equation}
    \label{eq:opt_power_scalar}
    P^* = \sigma_w \sqrt{\frac{(2a - 1)\sigma_m^2}{\rho}} - \sigma_w^2.
\end{equation}
Otherwise, the SE is non-informative ($P^* = 0$).
\end{enumerate}
\end{thm}

\begin{pf}
We analyze the game where $m \sim \N(0, \sigma_m^2)$ and $w \sim \N(0, \sigma_w^2)$. Let $D \coloneq \E[(m-u)^2]$ and $P \coloneq \E[x^2]$ denote the scalar mean squared estimation error and average transmission power, respectively.

\noindent\emph{Step 1: Encoder's cost.}
By Prop.~\ref{prop:decoder_mmse}, the decoder's best response is $u = \E[m|y]$. Following the same orthogonal decomposition logic used in Section~\ref{sec:cheaptalk_scalar}, we substitute $\E[u^2] = \sigma_m^2 - D$ into the encoder's cost function \eqref{eq:inst_cost_e}. The expected cost simplifies to
\begin{equation}
        J^{\mathrm{e}} =(2a - 1) D + \rho P + (a-1)^2 \sigma_m^2 + b^2.
        \label{eq:Je_DP}
\end{equation}
The encoder's problem is to minimize \eqref{eq:Je_DP} subject to the physical constraints imposed by the channel.
    
\noindent\emph{Step 2: Information-theoretic lower bound.}
First, we relate the mutual information to the distortion using differential entropy:
\begin{align}
I(m; y) &= h(m) - h(m|y) = h(m) - h(m - \mathbb{E}[m|y]\mid y) \nonumber\\
&\geq h(m) - h(m - \mathbb{E}[m|y]) \nonumber\\
&\overset{(a)}{\geq} \frac{1}{2} \log_2(2\pi e\sigma_m^2) - \frac{1}{2} \log_2(2\pi eD)
= \frac{1}{2} \log_2\left(\frac{\sigma_m^2}{D}\right).\nonumber 
\end{align}
Inverting this relation yields $D \geq \sigma_m^2 2^{-2I(m;y)}$. Using the data-processing inequality and the definition of channel capacity $C(P) = \sup_{p(x):\mathbb{E}[x^2] \leq P} I(x;y)$, we obtain:
\begin{align}
D &\overset{(b)}{\geq} \sigma_m^2 \,2^{-2I(x;y)} \nonumber\\
&\geq \sigma_m^2\, 2^{-2 C(P)}\overset{(c)}{=} \sigma_m^2\, 2^{-2\left[ \frac{1}{2} \log_2 \left(1+{P}/{\sigma_w^2}\right) \right]} \nonumber\\
\Rightarrow D &= \mathbb{E}[(m - u)^2]\geq \frac{\sigma_m^2}{1 + P/\sigma_w^2}. 
\label{eq:D_lower_bound}
\end{align}
Here, (a) holds since $h(m) = \frac{1}{2} \log_2(2\pi e\sigma_m^2)$ for a Gaussian source, (b) follows from the data-processing inequality, and (c) substitutes the Gaussian channel capacity.
    
\medskip
\noindent\emph{Step 3: Optimization.}
The encoder seeks to minimize the cost $J^{\mathrm{e}} = (2a - 1) D + \rho P + \mathrm{const}$. We analyze the two regimes for $a$ separately.
\begin{enumerate}[label=\textbf{\arabic*.},topsep=-0.5\baselineskip]
\item{Case $a \leq \tfrac{1}{2}$.} Then $(2a - 1)\leq 0$. To minimize $(2a-1)D$, the encoder must maximize the distortion $D$. The maximum MSE (error variance) is the prior variance $\sigma_m^2$, implying $D \leq \sigma_m^2$. Simultaneously, the encoder seeks to minimize the power cost $\rho P$ (since $\rho > 0$).
Both terms $(2a-1)D$ and $\rho P$ are minimized when the encoder transmits no information ($x=0$), resulting in $P=0$ and maximal distortion $D = \sigma_m^2$. Thus, the unique global minimum is at $P^*=0$.
\item{Case $a > \tfrac{1}{2}$.} The coefficient $(2a - 1)$ is positive. In this regime, the encoder faces a trade-off between reducing distortion and saving power. Substituting the lower bound \eqref{eq:D_lower_bound} into \eqref{eq:Je_DP} yields
\begin{equation}
    \label{eq:lower_bound_cost}
    J^{\mathrm{e}} \geq (a-1)^2 \sigma_m^2 + b^2 + \underbrace{\frac{(2a - 1)\sigma_m^2 \sigma_w^2}{\sigma_w^2 + P} + \rho P}_{\coloneq f(P)}.
\end{equation}
We minimize $f(P)$ over $P \geq 0$. Since $f''(P) > 0$, the function is strictly convex. Setting the derivative $f'(P) = \rho - \frac{(2a - 1)\sigma_m^2 \sigma_w^2}{(\sigma_w^2 + P)^2}$ to zero yields the unique unconstrained minimizer $P^*$ given in \eqref{eq:opt_power_scalar}.

Thus, an informative equilibrium exists if and only if $P^* > 0$. Imposing this inequality yields $\sqrt{\frac{(2a - 1)\sigma_m^2}{\rho}} > \sigma_w$, which simplifies to condition \eqref{eq:rho_condition}. If this condition is not met, $f'(0) \geq 0$ and the minimum occurs at the boundary $P^*=0$ (non-informative).
 \end{enumerate}

\medskip
\noindent\emph{Step 4: Achievability.}
The bound \eqref{eq:D_lower_bound} is tight for Gaussian sources over AWGN channels using linear encoding. Specifically, a linear policy $\gamma^{\mathrm{e}}(m) = \alpha m$ with $\alpha^2 = P^*/\sigma_m^2$ results in the lower bound MMSE $D = \frac{\sigma_m^2}{1 + P^*/\sigma_w^2}$.
Thus, the linear policy achieves the global lower bound of the cost function $J^{\mathrm{e}}$.
Consequently:
\begin{enumerate}[label=\textbf{\arabic*.},topsep=-0.5\baselineskip]  
\item If $a \leq \tfrac{1}{2}$, $P^*=0$, leading to $x=0$ (non-informative).
\item If condition \eqref{eq:rho_condition} holds, the optimal policy is the linear map corresponding to $P^* > 0$ (informative).\hfill $\qed$
\end{enumerate}
\end{pf}

\subsection{Multidimensional Signaling}
\label{sec:signaling_multidim}

We now extend the analysis to the multidimensional signaling game with a noisy channel and transmission cost. The source is $\boldsymbol{m} \sim \mathcal{N}(\boldsymbol{0}, \Sigma_m)$ with $\Sigma_m \succ O$, and the encoder has a general bias matrix $A \in \mathbb{R}^{n \times n}$.

By the same decomposition as in Section~\ref{sec:cheaptalk_multidim}, the encoder's expected cost can be written as $J^{\mathrm{e}} = \tr(V\Sigma_u) + \rho P + \mathrm{const}$, where $V = I - (A + A^\top)$ is the cost kernel from the cheap-talk analysis, $\Sigma_u = \mathbb{E}[\boldsymbol{u}\boldsymbol{u}^\top]$ is the posterior mean covariance, $P = \mathbb{E}[\|\boldsymbol{x}\|^2]$ is the transmission power, and $\mathrm{const} = \tr(A^\top A \Sigma_m) + \|\boldsymbol{b}\|^2$ is independent of the encoder policy. The key difference from the cheap-talk setting is that the achievable pairs $(\Sigma_u, P)$ are now constrained by the channel capacity.

First, following the approach from \cite{saritas2017quadratic}, we derive bounds on the achievable posterior covariance.

\begin{lem}[Determinant bound]
\label{lem:det_bound}
For any encoder policy with total transmission power $P = \mathbb{E}[\|\boldsymbol{x}\|^2]$, the error covariance $\Sigma_e = \Sigma_m - \Sigma_u$ satisfies
\begin{equation}
    |\Sigma_e| \geq |\Sigma_m| \, 2^{-2C_{\mathrm{tot}}(P)},
    \label{eq:det_bound}
\end{equation}
where $C_{\mathrm{tot}}(P)$ is the total capacity of the $n$-dimensional additive Gaussian noise channel with power constraint $P$.
\end{lem}

\begin{pf}
By the data-processing inequality and the definition of channel capacity,
\begin{align*}
    I(\boldsymbol{m}; \boldsymbol{y}) &= h(\boldsymbol{m}) - h(\boldsymbol{m}|\boldsymbol{y}) \\
    &= h(\boldsymbol{m}) - h(\boldsymbol{m} - \mathbb{E}[\boldsymbol{m}|\boldsymbol{y}] \mid \boldsymbol{y}) \\
    &\geq h(\boldsymbol{m}) - h(\boldsymbol{m} - \mathbb{E}[\boldsymbol{m}|\boldsymbol{y}]) \\
    &\geq \frac{1}{2}\log_2((2\pi e)^n |\Sigma_m|) - \frac{1}{2}\log_2((2\pi e)^n |\Sigma_e|) \\
    &= \frac{1}{2}\log_2\left(\frac{|\Sigma_m|}{|\Sigma_e|}\right).
\end{align*}
Since $I(\boldsymbol{m}; \boldsymbol{y}) \leq I(\boldsymbol{x}; \boldsymbol{y}) \leq C_{\mathrm{tot}}(P)$, we obtain $|\Sigma_e| \geq |\Sigma_m| \, 2^{-2C_{\mathrm{tot}}(P)}$.
\hfill $\qed$
\end{pf}

For the colored Gaussian noise channel with covariance $\Sigma_w$, we use the water-filling capacity and the arithmetic–geometric mean inequality to bound the capacity.
\begin{lem}[Capacity bound]
\label{lem:capacity_bound}
The capacity $C_{\mathrm{tot}}(P)$ of the additive Gaussian noise channel with covariance $\Sigma_w \succ O$ satisfies
\begin{equation}
    2^{-2{C_{\mathrm{tot}}(P)}/{n}} \geq \frac{|\Sigma_w|^{1/n}}{P/n + \frac{1}{n}\tr(\Sigma_w)}.
    \label{eq:capacity_bound}
\end{equation}
\end{lem}

\begin{pf}
The capacity of the additive colored Gaussian noise channel is achieved by the water-filling power allocation. The total capacity is given by \citep[Eq. (9.166)]{cover1999elements}:
\[
    C_{\mathrm{tot}}(P) = \sum_{i=1}^n \frac{1}{2} \log_2 \left( 1 + \frac{\max(\nu - \lambda_i, 0)}{\lambda_i} \right),
\]
where $\lambda_1, \dots, \lambda_n$ are the eigenvalues of $\Sigma_w$, and $\nu$ is the water level chosen such that the total power constraint is satisfied: $\sum_{i=1}^n \max(\nu - \lambda_i, 0) = P$.

We examine the term $2^{-2{C_{\mathrm{tot}}{(P)}}/{n}}$:
\begin{align*}
    2^{-2C_{\mathrm{tot}}(P)/n} &= 2^{-\frac{2}{n} \sum_{i=1}^n \frac{1}{2} \log_2 \left( 1 + \frac{\max(\nu - \lambda_i, 0)}{\lambda_i} \right)} \\
    &= \prod_{i=1}^n \left( 1 + \frac{\max(\nu - \lambda_i, 0)}{\lambda_i} \right)^{-1/n} \\
    &= \prod_{i=1}^n \left( \frac{\lambda_i}{\max(\nu, \lambda_i)} \right)^{1/n} \\
    &= \frac{(\prod_{i=1}^n \lambda_i)^{1/n}}{(\prod_{i=1}^n \max(\nu, \lambda_i))^{1/n}}.
\end{align*}
The numerator is the geometric mean of the eigenvalues, which is $|\Sigma_w|^{1/n}$. For the denominator, we apply the inequality of arithmetic and geometric means (AM--GM):
\begin{align*}
    2^{-2C_{\mathrm{tot}}(P)/n} &\stackrel{(a)}{\geq} \frac{|\Sigma_w|^{1/n}}{\frac{1}{n} \sum_{i=1}^n \max(\nu, \lambda_i)} \\
    &= \frac{|\Sigma_w|^{1/n}}{\frac{1}{n} \sum_{i=1}^n \left( \max(\nu - \lambda_i, 0) + \lambda_i \right)} \\
    &= \frac{|\Sigma_w|^{1/n}}{\frac{1}{n} \left( P + \tr(\Sigma_w) \right)}.
\end{align*}
Here, (a) follows as the geometric mean is less than or equal to the arithmetic mean.
\hfill $\qed$
\end{pf}

Combining Lemmas~\ref{lem:det_bound} and \ref{lem:capacity_bound} with the AM--GM inequality yields bounds on the achievable posterior covariance.
\begin{prop}[Posterior covariance bounds]
\label{prop:cov_bounds}
The trace of the posterior mean covariance $\Sigma_u$ is upper-bounded by
\begin{equation*}
    \tr(\Sigma_u) \leq \tr(\Sigma_m) - n^2 \, |\Sigma_m|^{\frac{1}{n}} |\Sigma_w|^{\frac{1}{n}} \left(P + \tr(\Sigma_w)\right)^{-1}.
    \label{eq:trace_upper_bound}
\end{equation*}
\end{prop}

\begin{pf}
By the AM--GM inequality applied to the diagonal entries of $\Sigma_e$,
\[
    \tr(\Sigma_e) = \sum_{i=1}^n \Sigma_e(i,i) \geq n \left(\prod_{i=1}^n \Sigma_e(i,i)\right)^{\frac{1}{n}} \stackrel{(*)} \geq n \, |\Sigma_e|^{\frac{1}{n}}.
\]
Where (*) follows from the Hadamard inequality (as $\Sigma_e$ is positive semi-definite). From Lemma~\ref{lem:det_bound}, we have $|\Sigma_e|^{\frac{1}{n}} \geq |\Sigma_m|^{\frac{1}{n}} 2^{-2\frac{C_{\mathrm{tot}}(P)}{n}}$. Substituting the bound from Lemma~\ref{lem:capacity_bound} gives
\[
    \tr(\Sigma_e) \geq n \, |\Sigma_m|^{\frac{1}{n}} \left(\frac{n|\Sigma_w|^{{\frac{1}{n}}}}{P + \tr(\Sigma_w)}\right) = \frac{n^2 |\Sigma_m|^{\frac{1}{n}} |\Sigma_w|^{\frac{1}{n}}}{P + \tr(\Sigma_w)}.
\]
The result follows from $\tr(\Sigma_u) = \tr(\Sigma_m) - \tr(\Sigma_e)$.
\hfill $\qed$
\end{pf}

We now analyze the possibility of informative equilibria. To maintain tractability, we consider the case of isotropic sensitivity $A = aI$.
\begin{thm}[Signaling threshold]
\label{thm:proportional_bias}
Consider $A = aI$ for scalar $a > 0$.
\begin{enumerate}[label=\textbf{\arabic*.},topsep=-1\baselineskip,itemindent=10pt]
\item If $a \leq \tfrac{1}{2}$, the unique equilibrium is non-informative for any $\rho > 0$.
\item If $a > \tfrac{1}{2}$, the information-theoretic lower bound on the encoder's cost, derived from Prop.~\ref{prop:cov_bounds}, is minimized by a strictly positive power $P > 0$ if and only if
    \begin{equation}
        \rho < (2a-1) \frac{n^2 |\Sigma_m|^{\frac{1}{n}} |\Sigma_w|^{\frac{1}{n}}}{(\tr(\Sigma_w))^2}.
        \label{eq:rho_threshold_multidim}
    \end{equation}
This condition characterizes the regime in which the information-theoretic constraints allow cost reduction via signaling.
\end{enumerate}
\end{thm}

\begin{pf}
With $V = (1-2a)I$, the encoder minimizes $J^{\mathrm{e}} = (1-2a)\tr(\Sigma_u) + \rho P$.
\begin{enumerate}[label=\textbf{\arabic*.},topsep=-1\baselineskip,itemindent=10pt]
\item{If $a \leq \tfrac{1}{2}$.} The coefficient $(1-2a)$ is non-negative. Minimizing the cost requires minimizing $\tr(\Sigma_u)$ and $P$. The global minimum is attained at $\Sigma_u = O$ and $P = 0$, which is achievable by transmitting no signal. Thus, the equilibrium is non-informative.
\item {If $a > \tfrac{1}{2}$.} The coefficient $(1-2a)$ is negative. Substituting the bound from Prop.~\ref{prop:cov_bounds} into the objective yields the lower bound function $f(P)$:
\[
    J^{\mathrm{e}}(P) \geq f(P) \coloneqq (1-2a) \left[ \tr(\Sigma_m) - \frac{\kappa}{P + \tau} \right] + \rho P,
\]
where $\kappa = n^2 |\Sigma_m|^{\frac{1}{n}} |\Sigma_w|^{\frac{1}{n}}$ and $\tau = \tr(\Sigma_w)$. We minimize $f(P)$ over $P \geq 0$. Computing the derivatives:
\[
    f'(P) = -\frac{(2a-1)\kappa}{(P+\tau)^2} + \rho, \quad f''(P) = \frac{2(2a-1)\kappa}{(P+\tau)^3}.
\]
Since $a > 1/2$, $f''(P) > 0$ for all $P \geq 0$, so $f(P)$ is strictly convex. The unique global minimum occurs at $P^* > 0$ if and only if $f'(0) < 0$, which is equivalent to
\[
    \rho < \frac{(2a-1)\kappa}{\tau^2} = (2a-1) \frac{n^2 |\Sigma_m|^{\frac{1}{n}} |\Sigma_w|^{\frac{1}{n}}}{(\tr(\Sigma_w))^2}.
\]
Thus, condition \eqref{eq:rho_threshold_multidim} is necessary and sufficient for the lower bound $f(P)$ to be minimized at a non-zero power.\hfill$\qed$
\end{enumerate}
\end{pf}

\begin{cor}[The i.i.d.\ case]
    \label{cor:iid_case}
    Consider the case where $\Sigma_m = \sigma_m^2 I$ and $\Sigma_w = \sigma_w^2 I$. In this setting, an informative affine equilibrium exists if and only if
    \begin{equation}
        \rho < (2a-1) \frac{\sigma_m^2}{\sigma_w^2}.
        \label{eq:rho_threshold_iid}
    \end{equation}
    recovering the scalar threshold.
\end{cor}

\begin{pf}
In the i.i.d.\ case, $\tr(\Sigma_m) = n\sigma_m^2$ and $|\Sigma_m|^{1/n} = \sigma_m^2$, so $\tr(\Sigma_m) = n|\Sigma_m|^{1/n}$. Similarly, $\tr(\Sigma_w) = n|\Sigma_w|^{1/n}$. Consequently, the AM--GM inequalities used in Lemmas~\ref{lem:capacity_bound} and Prop.~\ref{prop:cov_bounds} hold with equality at $P=0$ and for uniform power allocations. Specifically, $f(0) = J^{\mathrm{e}}(0)$, and the lower bound $f(P)$ is achievable by scalar linear policies $\boldsymbol{x} = \alpha \boldsymbol{m}$. Therefore, the condition \eqref{eq:rho_threshold_multidim}, which simplifies to \eqref{eq:rho_threshold_iid} in this case, becomes necessary and sufficient for the existence of an informative equilibrium.
\hfill $\qed$
\end{pf}

\section{Numerical Illustrations}
\label{sec:simulations}
We illustrate our theoretical findings through numerical examples covering both the scalar signaling game and the multidimensional cheap-talk setting.\\
\textbf{Scalar Case.} Figure~\ref{fig:phase_diagram} shows the phase transition between informative and non-informative equilibria in the $(a, \rho)$ parameter space. We set $\sigma_m^2 = 1$ and $\sigma_w^2 = 0.5$. The solid black curve represents the theoretical boundary $\rho = \frac{\sigma_m^2}{\sigma_w^2}(2a-1)$. For $a \leq 1/2$ (left of the dashed line) or when the cost is too high (hatched region), the equilibrium is strictly non-informative, resulting in zero transmission power ($P^*=0$). For $a > 1/2$, an informative equilibrium exists below the boundary. The color gradient indicates the magnitude of the optimal transmission power $P^*$, showing that communication intensity increases as the sensitivity mismatch decreases (higher $a$) or the cost $\rho$ decreases.
\begin{figure}[t!]
    \centering
    \includegraphics[width=0.8\linewidth]{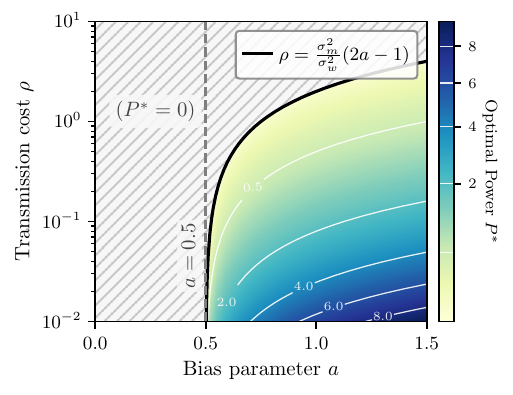}
    \caption{Phase diagram for the scalar signaling game. The boundary separates the region where communicating is beneficial for the encoder (informative) from the region where the optimal encoder is non-informative. The color intensity represents the optimal power $P^*$.}
    \label{fig:phase_diagram}
\end{figure}

\textbf{Multidimensional Case.} To visualize the spectral characterization of information revelation Theorem~\ref{thm:SE_covariance}, we consider a 2-dimensional cheap-talk game (noiseless, $\rho=0$). We set the sensitivity matrix $A = \mathrm{diag}(0.8, 0.2)$ and assume a zero-mean source with marginal variances $\Sigma_{11}=1$ and $\Sigma_{22}=1.5$.
Figure~\ref{fig:vector_scatter} compares the equilibrium outcomes for an independent source ($\Sigma_{12}=0$) versus a correlated source ($\Sigma_{12}=0.3$).
In panel (a) (independent source), the encoder fully reveals $m_1$ but completely suppresses $m_2$, confirming that the encoder filters out directions with high mismatch.\\
In panel (b) (correlated source), the encoder projects onto the subspace spanned by the eigenvector of the transformed cost matrix $B$ corresponding to the negative eigenvalue. Due to the source correlation, this optimal signaling direction is rotated relative to the canonical axes (deviating from the first dimension). Consequently, the decoder infers information about $m_2$ through this rotated projection, resulting in the estimate $u_2$ being correlated with $m_1$ in a manner that optimally balances information revelation against the sensitivity mismatch.
\begin{figure}[t!]
    \centering
    \includegraphics[width=0.95\linewidth]{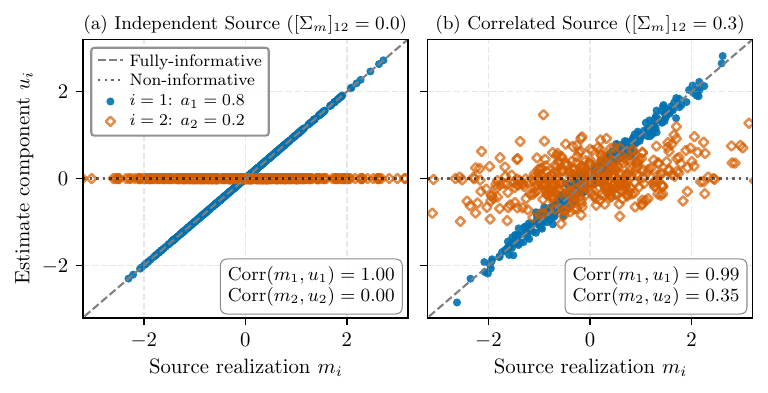}
    \caption{Illustration of equilibrium behavior for 2D cheap talk with $A=\mathrm{diag}(0.8, 0.2)$. (a) \textbf{Independent Source:} The problem decouples; Component 1 is revealed, while Component 2 is suppressed ($u_2=0$). (b) \textbf{Correlated Source:} The encoder still doesn't reveal $m_2$, but the decoder infers partial information about $m_2$ via its correlation with the revealed $m_1$, showing the interaction between mismatch geometry and source prior.}
    \label{fig:vector_scatter}
\end{figure}

\section{Conclusion and Perspectives}
\label{sec:conclusion}
This paper formulated and solved a Gaussian signaling game driven by a linear sensitivity mismatch between an encoder and a decoder. Under a Stackelberg commitment model, we derived explicit conditions determining when communication occurs. The analysis shows that a linear bias restricts information disclosure. The encoder only transmits data along specific eigenspaces of the mismatch matrix, and communication collapses if the conflict of interest crosses a critical threshold. Future research includes analyzing Nash equilibria without commitment, extending the noisy channel results to non-commutative settings, and embedding this model into dynamic control systems with an evolving sensitivity mismatch.

\bibliography{ifc}             
\end{document}